\newcommand{\shortversion}[1]{}
\newcommand{\longversion}[1]{#1}

\shortversion{
\documentclass{llncs} 
\title{Parameterized Complexity and Kernel Bounds for Hard Planning Problems}
\author{
Christer B\"{a}ckstr\"{o}m\inst{1} \and
Peter Jonsson\inst{1} \and
Sebastian Ordyniak\inst{2} \and
Stefan~Szeider\inst{3}\thanks{Research supported by the ERC, grant
  reference 239962.}
}
\institute{
Department of Computer Science, Link{\"o}ping University, Link{\"o}ping, Sweden \and
Faculty of Informatics, Masaryk University, Brno, Czech Republic \and
Institute of Information Systems, Vienna University of Technology,
Vienna, Austria 
}

\usepackage{enumitem,calc}
}

\longversion{
\documentclass[10pt,usletter]{article}
\usepackage[totalwidth=450pt,totalheight=640pt]{geometry}
\title{Parameterized Complexity and Kernel Bounds for\\ Hard Planning Problems}
\author{
Christer B\"{a}ckstr\"{o}m$^1$,
Peter Jonsson$^1$,
Sebastian Ordyniak$^2$, and
Stefan Szeider$^3$\thanks{Research supported by the ERC, grant
  reference 239962.}\\[0.1cm]
\mbox{}\small$^1$Department of Computer Science, Link{\"o}ping University,
Link{\"o}ping, Sweden\\[-4pt]
\small christer.backstrom@liu.se, peter.jonsson@liu.se\\
\mbox{}\small$^2$Faculty of Informatics, Masaryk University, Brno, Czech Republic\\[-4pt] 
\small sordyniak@gmail.com\\
\mbox{}\small$^3$Institute of Information Systems, Vienna University of Technology,
Vienna, Austria\\[-4pt]
\small  stefan@szeider.net
}

\date{}

\usepackage{amsthm}
}

\usepackage{url} 

\usepackage{amsmath}
\usepackage{amssymb}
\usepackage{booktabs}
\usepackage{tikz}
\usepackage{multirow}
\usepackage{enumitem}
\longversion{
  \usepackage{times}
}

\newcommand{\hy}{\hbox{-}\nobreak\hskip0pt}

\newcommand{\SB}{\{\,}%
\newcommand{\SM}{\;{:}\;}%
\newcommand{\SE}{\,\}}%
\newcommand{\Card}[1]{|#1|}

\newcommand{\NP}{\text{\normalfont NP}}
\newcommand{\coNP}{\text{\normalfont co-NP}}
\newcommand{\FPT}{\text{\normalfont FPT}}

\newcommand{\W}[1]{\text{\normalfont W[#1]}}

\newtheorem{LEM}{Lemma} 
\newtheorem{THE}{Theorem} 
\newtheorem{PRO}{Proposition} 
 
\newtheorem{CLM}{Claim}

\shortversion{

\renewenvironment{proof}{\vspace{-3mm}\begin{pf}}{\qed\end{pf}}
 }

\let\phi=\varphi
\newcommand{\concat}{\circ}

\newcommand{\undv}{\mathbf{u}}

\newcommand{\NOPOLYKERNEL}{\coNP \subseteq \NP/\textup{poly}}

\renewcommand{\phi}{\varphi}

\renewcommand{\emptyset}{\varnothing}  
\newcommand{\union}{\cup} 		

\newcommand{\card}[1]{{|#1|}}		
\newcommand{\set}[1]{\{{#1}\}}          

\newcommand{\tuple}[1]{\langle{#1}\rangle}  
\newcommand{\seq}[1]{\langle #1 \rangle}

\newcommand{\instance}[1]{{\mathbb{#1}}}

\newcommand{\cc}[1]{{\mbox{\textnormal{#1}}}}  

\newcommand{\poly}{\cc{P}}
\newcommand{\PSPACE}{\cc{\textsc{Pspace}}}

\newcommand{\insti}{\instance{I}}
\newcommand{\iplan}{\instance{P}}      

\newcommand{\strips}{{\textsc{Strips}}}
\newcommand{\sasplus}{{SAS$^+$}}

\newcommand{\vars}{V} 
\newcommand{\dom}{D}  
\newcommand{\acts}{A}  
\newcommand{\init}{I}  
\newcommand{\goal}{G}  
\newcommand{\pre}{\mathrm{pre}}  
\newcommand{\eff}{\mathrm{eff}}  %

\newcommand{\proj}[2]{{#1[{#2}]}}

\newcommand{\undef}{\mathbf{u}}

\newcommand{\plan}{\omega}

\newcommand{\BPE}{\textsc{Bounded Planning}}

\begin{document}

\maketitle

\begin{abstract}\sloppypar 
  The {\em propositional planning} problem is a notoriously difficult
  computational problem.  Downey et al.~(1999) initiated the
  parameterized analysis of planning (with plan length as the
  parameter) and B\"{a}ckstr\"{o}m et al.~(2012) picked up this line
  of research and provided an extensive parameterized analysis under
  various restrictions, leaving open only one stubborn case.  We
  continue this work and provide a full classification.  In
  particular, we show that the case when actions have no preconditions
  and at most $e$ postconditions is fixed-parameter tractable if
  $e\leq 2$ and W[1]-complete otherwise. We show fixed-parameter
  tractability by a reduction to a variant of the Steiner Tree
  problem; this problem has been shown fixed-parameter
  tractable by Guo et al.~(2007).  If a problem is fixed-parameter
  tractable, then it admits a polynomial-time self-reduction to
  instances whose input size is bounded by a function of the
  parameter, called the {\em kernel}.  For some problems, this
  function is even polynomial which has desirable computational
  implications. Recent research in parameterized complexity has
  focused on classifying fixed-parameter tractable problems on whether
  they admit polynomial kernels or not. We revisit all the previously
  obtained restrictions of planning that are fixed-parameter tractable
  and show that none of them admits a polynomial kernel unless the
  polynomial hierarchy collapses to its third level.

\end{abstract}

\longversion{
  \pagestyle{plain}
  \thispagestyle{empty}
}
\shortversion{
}

\section{Introduction}

The propositional planning problem has been the subject of intensive
study in knowledge representation, artificial intelligence and control
theory and is relevant for a large number of industrial
applications~\cite{GhallabNauTraverso04}.  The problem involves deciding
whether an \emph{initial state}---an $n$-vector over some set $D$–--can
be transformed into a \emph{goal state} via the application of {\em
  operators} each consisting of {\em preconditions} and {\em
  post-conditions} (or {\em effects}) stating the conditions that need to hold before the
operator can be applied and which conditions will hold after the
application of the operator, respectively.
It is
known that deciding whether an instance has a solution is
\PSPACE-complete, and it remains at least NP-hard under various
restrictions~\cite{Bylander94,BackstromNebel95}.
In view of this intrinsic difficulty of the problem, it is natural to
study it within the framework of Parameterized Complexity which offers
the more relaxed notion of \emph{fixed-parameter tractability} (FPT).  A
problem is fixed-parameter tractable if it can be solved in time
$f(k)n^{O(1)}$ where $f$ is an arbitrary function of the parameter and
$n$ is the input size. Indeed, already in a 1999 paper, Downey, Fellows
and Stege \cite{DowneyFellowsStege99} initiated the parameterized
analysis of propositional planning, taking the minimum number of steps
from the initial state to the goal state (i.e., the length of the
solution plan) as the parameter; this is also the parameter used
throughout this paper.  More recently, B\"{a}ckstr\"{o}m et
al.~\cite{BackstromChenJonssonOrdyniakSzeider12} picked up this line of
research and provided an extensive analysis of planning under various
syntactical restrictions, in particular the syntactical restrictions
considered by Bylander~\cite{Bylander94} and by B\"{a}ckstr\"{o}m and
Nebel~\cite{BackstromNebel95}, leaving open only one stubborn class of
problems  where
operators have no preconditions but may involve up to~$e$
postconditions (effects).
 
\subsection*{New Contributions}
We provide a full parameterized complexity analysis of propositional
planning without preconditions. In particular, we show the following
dichotomy:
\begin{enumerate}[label=(1),leftmargin=*]
\item Propositional planning where operators have no preconditions but may
  have up to~$e$ postconditions is fixed-parameter tractable for $e\le
  2$ and $\W{1}$\hy complete for $e>2$.
\end{enumerate}
$\W{1}$ is a parameterized complexity class of problems
that are believed to be not fixed-parameter tractable. Indeed, the
fixed-parameter tractability of a $\W{1}$-complete problem implies that
the Exponential Time Hypothesis
fails~\cite{ChenHuangKanjXia06,FlumGrohe06}.  We establish the hardness
part of the dichotomy (1) by a reduction from a variant of the
$k$-\textsc{Clique} problem.  The case $e=2$ is known to be $\NP$-hard~\cite{Bylander94}.
Its difficulty comes from the fact that possibly one of the two
postconditions might set a variable to its desired value, but the other
postcondition might change a variable from a desired value to an
undesired one. This can cause a chain of operators so that finally all
variables have their desired value. We show that this behaviour can be
modelled by means of a certain problem on Steiner trees in directed
graphs, which was recently shown to be fixed-parameter tractable by Guo,
Niedermeier and Suchy~\cite{GuoNiedermeierSuchy11}.  We would like to
point out that this case (0 preconditions, 2 postconditions) is the
only fixed-parameter tractable case among the NP-hard cases in Bylander's
system of restrictions (see Table~\ref{table:bylander}).
  
\medskip Our second set of results is concerned with bounds on problem
kernels for planning problems. It is known that a decidable problem is
fixed-parameter tractable if and only if it admits a polynomial-time
self-reduction where the size of the resulting instance is bounded by a
function $f$ of the parameter~\cite{Fellows06,GuoNiedermeier07,Fomin10}.
The function $f$ is called the \emph{kernel size}.  By providing upper
and lower bounds on the kernel size, one can rigorously establish the
potential of polynomial-time preprocessing for the problem at hand.
Some NP-hard combinatorial problems such as $k$-\textsc{Vertex Cover}
admit polynomially sized kernels, for others such as $k$-\textsc{Path}
an exponential kernel is the best one can hope
for~\cite{BodlaenderDowneyFellowsHermelin09}.  We examine all planning
problems that we have previously been shown to be fixed-parameter tractable
on whether they admit polynomial kernels. Our results are negative throughout.
In particular, it is unlikely that the FPT part in the above
dichotomy~(1) can be improved to a polynomial kernel:
\begin{enumerate}[label=(2),leftmargin=*]
\item Propositional planning where operators have no preconditions
  but may have up to~2 postconditions does not admit a polynomial
  kernel unless $\NOPOLYKERNEL$.
\end{enumerate}
Recall that by Yap's Theorem \cite{Yap83} $\NOPOLYKERNEL$ implies the
(unlikely) collapse of the Polynomial Hierarchy to its third level.  We
establish the kernel lower bound by means of the technique of
\emph{OR-compositions}~\cite{BodlaenderDowneyFellowsHermelin09}.  We
also consider the ``PUBS'' fragments of planning as introduced by
B\"{a}ckstr\"{o}m and Klein~\cite{BackstromKlein91}.  These fragments
arise under combinations of syntactical properties (postunique~(P),
unary~(U), Boolean~(B), and single-valued~(S); definitions are provided in
Section~\ref{sec:planning-framework}).

\begin{enumerate}[label=(3),leftmargin=*]
\item None of the fixed-parameter tractable but NP-hard PUBS
  restrictions of propositional planning admits a polynomial kernel,
  unless $\NOPOLYKERNEL$.
\end{enumerate}
\begin{table}[tb]
\centering
\shortversion{\small}
\begin{tabular}{@{}l@{~~~~~}llll@{}} 
  \toprule
  & $e=1$     & $e=2$ & fixed $e > 2$ ~~~ & arbitrary $e$\\
    \midrule
  $p=0$       & in \poly   & in \FPT${}^*$      & \W{1}-C${}^*$  & \W{2}-C${}$ \\ 
  & in \poly    & \NP-C      &\NP-C             & \NP-C \\ 
  \midrule
  $p=1$       & \W{1}-C    & \W{1}-C & \W{1}-C   & \W{2}-C \\
  & \NP-H       & \NP-H      &  \NP-H    & \PSPACE-C \\ 
    \midrule
  fixed $p > 1$ & \W{1}-C    & \W{1}-C & \W{1}-C   & \W{2}-C \\
  & \NP-H       & \PSPACE-C  & \PSPACE-C & \PSPACE-C \\ 
    \midrule
  arbitrary $p$    & \W{1}-C    & \W{1}-C & \W{1}-C   & \W{2}-C \\ 
  & \PSPACE-C   & \PSPACE-C  & \PSPACE-C   & \PSPACE-C \\ 
  \bottomrule
  \end{tabular}
\medskip

  \caption{Complexity of \BPE,
    restricting the number of preconditions ($p$) and 
    effects ($e$).
    The problems in FPT do not admit polynomial kernels.
    Results marked with * are obtained in this paper. All other
    parameterized results are
    from~\cite{BackstromChenJonssonOrdyniakSzeider12} and all
    classical results are
    from~\cite{Bylander94}.
}

    \label{table:bylander}
\end{table}
According to the PUBS lattice (see Figure~\ref{fig:pubs-lattice}),
only the two maximal restrictions PUB and PBS need to be considered.
Moreover, we observe from previous results that a polynomial kernel
for restriction PBS implies one for restriction PUB. Hence this leaves
restriction PUB as the only one for which we need to show a
super-polynomial kernel bound. We establish the latter, as above, by 
using  OR-compositions.

\begin{figure}[thb]
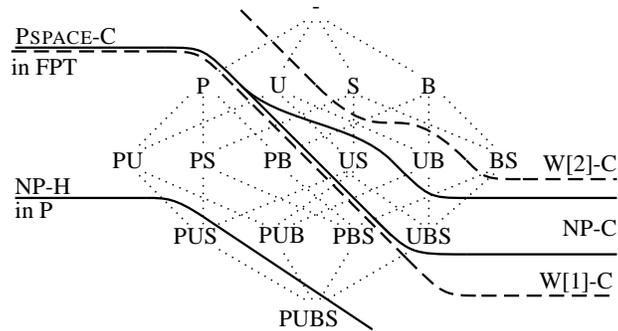

\centering
  \begin{pgfpicture}{-2cm}{-3cm}{2cm}{2cm}
    \newdimen\pubsdim
    \pgfextractx{\pubsdim}{\pgfpoint{5mm}{0mm}}
    \pgfsetxvec{\pgfpoint{5mm}{0mm}}
    \pgfsetyvec{\pgfpoint{0mm}{5mm}}

    \small
    \pgfsetlinewidth{0.2mm}
    \pgfsetdash{{0.2mm}{0.8mm}}{0mm}

    \pgfputat{\pgfxy(0,4)}{\pgfbox[center,center]{-}}

    \pgfline{\pgfxy(-0.3,3.7)}{\pgfxy(-2.7,2.3)} 
    \pgfline{\pgfxy(-0.1,3.7)}{\pgfxy(-0.9,2.4)} 
    \pgfline{\pgfxy(0.1,3.7)}{\pgfxy(0.9,2.4)}   
    \pgfline{\pgfxy(0.3,3.7)}{\pgfxy(2.7,2.3)}   

    \pgfputat{\pgfxy(-3,2)}{\pgfbox[center,center]{P}}
    \pgfputat{\pgfxy(-1,2)}{\pgfbox[center,center]{U}}
    \pgfputat{\pgfxy(1,2)}{\pgfbox[center,center]{S}}
    \pgfputat{\pgfxy(3,2)}{\pgfbox[center,center]{B}}

    \pgfline{\pgfxy(-3.4,1.7)}{\pgfxy(-4.8,0.3)} 
    \pgfline{\pgfxy(-3.0,1.7)}{\pgfxy(-3.0,0.3)} 
    \pgfline{\pgfxy(-2.6,1.7)}{\pgfxy(-1.2,0.3)} 

    \pgfline{\pgfxy(-1.4,1.7)}{\pgfxy(-4.6,0.3)} 
    \pgfline{\pgfxy(-0.9,1.7)}{\pgfxy(0.6,0.3)}  
    \pgfline{\pgfxy(-0.6,1.7)}{\pgfxy(2.6,0.3)}  

    \pgfline{\pgfxy(0.8,1.7)}{\pgfxy(-2.6,0.3)}  
    \pgfline{\pgfxy(1.0,1.7)}{\pgfxy(1.0,0.3)}   
    \pgfline{\pgfxy(1.2,1.7)}{\pgfxy(4.6,0.3)}   

    \pgfline{\pgfxy(2.8,1.7)}{\pgfxy(-0.6,0.3)}  
    \pgfline{\pgfxy(3.0,1.7)}{\pgfxy(3.0,0.3)}   
    \pgfline{\pgfxy(3.4,1.7)}{\pgfxy(4.8,0.3)}   

    \pgfputat{\pgfxy(-5,0)}{\pgfbox[center,center]{PU}}
    \pgfputat{\pgfxy(-3,0)}{\pgfbox[center,center]{PS}}
    \pgfputat{\pgfxy(-1,0)}{\pgfbox[center,center]{PB}}
    \pgfputat{\pgfxy(1,0)}{\pgfbox[center,center]{US}}
    \pgfputat{\pgfxy(3,0)}{\pgfbox[center,center]{UB}}
    \pgfputat{\pgfxy(5,0)}{\pgfbox[center,center]{BS}}
 
    \pgfline{\pgfxy(-3.4,-1.7)}{\pgfxy(-4.8,-0.3)} 
    \pgfline{\pgfxy(-3.0,-1.7)}{\pgfxy(-3.0,-0.3)} 
    \pgfline{\pgfxy(-2.6,-1.7)}{\pgfxy(0.8,-0.3)} 

    \pgfline{\pgfxy(-1.4,-1.7)}{\pgfxy(-4.6,-0.3)} 
    \pgfline{\pgfxy(-1.0,-1.7)}{\pgfxy(1.0,-0.3)}  
    \pgfline{\pgfxy(-0.6,-1.7)}{\pgfxy(2.6,-0.3)}  

    \pgfline{\pgfxy(0.6,-1.7)}{\pgfxy(-2.6,-0.3)}  
    \pgfline{\pgfxy(0.9,-1.7)}{\pgfxy(-0.8,-0.3)}  
    \pgfline{\pgfxy(1.2,-1.7)}{\pgfxy(4.6,-0.3)}   

    \pgfline{\pgfxy(2.8,-1.7)}{\pgfxy(0.8,-0.3)}  
    \pgfline{\pgfxy(3.0,-1.7)}{\pgfxy(3.0,-0.3)}  
    \pgfline{\pgfxy(3.2,-1.7)}{\pgfxy(4.8,-0.3)}  

    \pgfputat{\pgfxy(-3.2,-2)}{\pgfbox[center,center]{PUS}}
    \pgfputat{\pgfxy(-0.9,-2)}{\pgfbox[center,center]{PUB}}
    \pgfputat{\pgfxy(1,-2)}{\pgfbox[center,center]{PBS}}
    \pgfputat{\pgfxy(3,-2)}{\pgfbox[center,center]{UBS}}

    \pgfline{\pgfxy(-0.3,-3.7)}{\pgfxy(-2.7,-2.3)} 
    \pgfline{\pgfxy(-0.1,-3.7)}{\pgfxy(-0.9,-2.4)} 
    \pgfline{\pgfxy(0.1,-3.7)}{\pgfxy(0.9,-2.4)}   
    \pgfline{\pgfxy(0.3,-3.7)}{\pgfxy(2.7,-2.3)}   

    \pgfputat{\pgfxy(-0.2,-4.2)}{\pgfbox[center,center]{PUBS}}


    \pgfsetlinewidth{0.3mm}
    \pgfsetdash{}{0mm}

    \pgfxyline(-8.0,-1.0)(-4.5,-1.0)
    \pgfxycurve(-4.5,-1.0)(-4.0,-1.0)(-3.75,-1.0)(-3.0,-1.5)
    \pgfxyline(-3.0,-1.5)(1.5,-4.5)
    \pgfstroke
    \pgfputat{\pgfxy(-8.0,-1.1)}{\pgfbox[left,top]{in \poly}}

    \pgfxyline(-8.0,3.0)(-4.0,3.0)
    \pgfxycurve(-4.0,3.0)(-3.0,3.0)(-3.0,3.0)(-2.0,2.0)
    \pgfxyline(-2.0,2.0)(2.0,-2.0)
    \pgfxycurve(1.5,-1.5)(2.5,-2.5)(2.5,-2.5)(4.0,-2.5)
    \pgfxyline(4.0,-2.5)(8.0,-2.5)
    \pgfputat{\pgfxy(-8.0,-0.9)}{\pgfbox[left,bottom]{\NP-H}}
    \pgfputat{\pgfxy(8.0,-1.8)}{\pgfbox[right,center]{\NP-C}}
    \pgfxycurve(-2.0,2.0)(-1.0,1.0)(1.0,1.0)(2.0,0.0)
    \pgfxycurve(2.0,0.0)(3.0,-1.0)(3.0,-1.0)(4.0,-1.0)
    \pgfxyline(4.0,-1.0)(8.0,-1.0)
    \pgfstroke
    \pgfputat{\pgfxy(-8.0,3.1)}{\pgfbox[left,bottom]{\PSPACE-C}}
    \pgfsetlinewidth{0.3mm}
    \pgfsetdash{{2mm}{1mm}}{0mm}

    \begin{pgftranslate}{\pgfxy(-0.1,-0.1)}
      \pgfxyline(-8.0,3.0)(-4.0,3.0)
      \pgfxycurve(-4.0,3.0)(-3.0,3.0)(-3.0,3.0)(-2.0,2.0)
      \pgfxyline(-2.0,2.0)(2.5,-2.5)
      \pgfxycurve(2.5,-2.5)(3.5,-3.5)(3.5,-3.5)(4.5,-3.5)
      \pgfxyline(4.5,-3.5)(8.0,-3.5)
      \pgfstroke
      \pgfputat{\pgfxy(-8.0,2.8)}{\pgfbox[left,top]{in \FPT}}
      \pgfputat{\pgfxy(8.0,-3.4)}{\pgfbox[right,bottom]{\W{1}-C}}
    \end{pgftranslate}

    \begin{pgftranslate}{\pgfxy(0.0,0.0)}
      \pgfxyline(-2.0,4.0)(0.0,2.0)
      \pgfxycurve(0.0,2.0)(2.0,0.0)(2.0,2.0)(4.0,0.0)
      \pgfxycurve(4.0,0.0)(4.5,-0.5)(4.5,-0.5)(5.5,-0.5)
      \pgfxyline(5.5,-0.5)(8.0,-0.5)
      \pgfstroke
      \pgfputat{\pgfxy(8.0,-0.4)}{\pgfbox[right,bottom]{\W{2}-C}}
    \end{pgftranslate}
  \end{pgfpicture}
  \caption{Complexity of \BPE\ for the restrictions P, U, B and S illustrated as a lattice defined by all
    possible combinations of these restrictions~\cite{BackstromChenJonssonOrdyniakSzeider12}.
    As shown in this paper, PUS and PUBS are
    the only restrictions that admit a polynomial kernel, unless the
    Polynomial Hierarchy collapses. }
  \label{fig:pubs-lattice}
\end{figure}

\shortversion{The full proofs of statements marked with $\star
$
  are omitted due to space restrictions and can be found at \url{http://arxiv.org/abs/1211.0479}.}

\section{Parameterized Complexity}

We define the basic notions of Parameterized Complexity and
refer to other sources~\cite{DowneyFellows99,FlumGrohe06} 
for an in-depth treatment. 
A \emph{parameterized problem} is a set of pairs 
$\tuple{\insti,k}$,
the \emph{instances}, where $\insti$ is the main part and $k$ 
the \emph{parameter}. The parameter is usually a non-negative integer.
A parameterized problem is \emph{fixed-parameter tractable (FPT)} if
there exists an algorithm that solves any instance $\tuple{\insti,k}$ of
size $n$ in time $f(k)n^{c}$ where $f$ is an arbitrary computable
function and $c$ is a constant independent of both $n$ and $k$. 
\FPT\ is the class of all fixed-parameter
tractable decision problems.

Parameterized complexity offers a completeness theory, similar
to the theory of NP-completeness, that allows the accumulation of
strong theoretical evidence that some parameterized problems
are not fixed-parameter tractable. This theory is based on a
hierarchy of complexity classes
$\FPT \subseteq \W{1} \subseteq \W{2} \subseteq \cdots$
where all inclusions are believed to be strict. 
An \emph{fpt-reduction} from a parameterized problem $P$ to a
parameterized problem $Q$
if is a mapping $R$ from instances of $P$ to instances of $Q$ such
that
(i)~$\tuple{\insti,k}$ is a {\sc Yes}-instance of $P$ if and only if $\tuple{\insti',k'}=R(\insti,k)$
is a {\sc Yes}-instance of $Q$,
(ii)~there is a computable function $g$ such that $k' \leq g(k)$, and
(iii) there is a computable function $f$ and a constant $c$ such that $R$ can be
computed in time $O(f(k) \cdot n^c)$, where $n$ denotes the size
of $\tuple{\insti,k}$.

A \emph{kernelization}~\cite{FlumGrohe06} for a parameterized
problem~$P$ is an algorithm that takes an instance $\tuple{\insti,k}$ of
$P$ and maps it in time polynomial in $\Card{\insti}+k$ to an instance
$\tuple{\insti',k'}$ of $P$ such that $\tuple{\insti,k}$ is a
\textsc{Yes}-instance if and only if $\tuple{\insti',k'}$ is a
\textsc{Yes}-instance and $|\insti'|$ is bounded by some function
$f$~of~$k$.  The output $\insti'$ is called a \emph{kernel}. We say $P$
has a \emph{polynomial kernel} if $f$ is a polynomial.  Every
fixed-parameter tractable problem admits a kernel, but not necessarily a
polynomial kernel.

An \emph{OR-composition algorithm} for a parameterized problem $P$ maps
$t$ instances $\tuple{\insti_1,k},\dotsc,\tuple{\insti_t,k}$ of $P$ to
one instance $\tuple{\insti',k'}$ of $P$ such that the algorithm runs in
time polynomial in $\sum_{1 \leq i \leq t}|\insti_i|+k$, the parameter
$k'$ is bounded by a polynomial in the parameter $k$, and
$\tuple{\insti',k'}$ is a \textsc{Yes}-instance if and only if there is
an $1 \leq i \leq t$ such that $\tuple{\insti_i,k}$ is a
\textsc{Yes}-instance.
  
\begin{PRO}[Bodlaender, et al.~\cite{BodlaenderDowneyFellowsHermelin09}]\label{pro:or-comp-no-poly-kernel}
  If a parameterized problem $P$ has an OR\hy composition algorithm,
  then it has no polynomial kernel unless $\NOPOLYKERNEL$.
\end{PRO}
A \emph{polynomial parameter reduction} from a parameterized problem $P$
to a parameterized problem $Q$ is an fpt-reduction $R$ from $P$ to $Q$
such that (i)~$R$ can be computed in polynomial time (polynomial in
$\Card{\insti}+k)$, and (ii) there is a polynomial $p$ such that $k'\leq
p(k)$ for every instance $\tuple{\insti,k}$ of $P$ with
$\tuple{\insti',k'}=R(\tuple{\insti,k})$.  The \emph{unparameterized
  version} $\tilde{P}$ of a parameterized problem $P$ has the same
\textsc{YES} and \textsc{NO}-instances as $P$, except that the parameter
$k$ is given in unary~$1^k$.

\begin{PRO}[Bodlaender, Thomasse, and Yeo~\cite{BodlaenderThomasseYeo09}]\label{pro:poly-par-reduction}
  Let $P$ and $Q$ be two parameterized problems such that there is a
  polynomial parameter reduction from $P$ to $Q$, and assume that
  $\tilde{P}$ is \NP-complete and $\tilde{Q}$ is in \NP{}. Then, if $Q$
  has a polynomial kernel also $P$ has a polynomial kernel.
\end{PRO}

\section{Planning Framework}\label{sec:planning-framework}

We will now introduce the \sasplus\ formalism for specifying 
propositional planning problems~\cite{BackstromNebel95}.
We note that the propositional \strips\ language can be treated as 
the special case of \sasplus\ satisfying restriction B (which will be
defined below).
More precisely, this corresponds to the variant of \strips\ 
that allows negative preconditions; this formalism is often referred
to as {\sc Psn}.

Let $\vars = \set{v_1,\ldots,v_n}$ be a finite set of
\emph{variables} over a finite \emph{domain} $\dom$.
Implicitly define $\dom^+ = \dom \union \set{\undef}$,
where  $\undef$ is a special value (the \emph{undefined value}) not present in $\dom$.
Then  $\dom^n$ is the set of \emph{total states}
and $(\dom^+)^n$ is the set of \emph{partial states}
over $\vars$ and $\dom$, where $\dom^n \subseteq (\dom^+)^n$.
The value of a variable $v$ in a state $s \in (\dom^+)^n$
is denoted $\proj{s}{v}$.
A \emph{\sasplus\ instance} is a tuple
$\iplan = \tuple{\vars,\dom,\acts,\init,\goal}$
where  $\vars$ is a set of variables, 
$\dom$ is a domain,
$\acts$ is a set of \emph{actions},
$\init \in \dom^n$ is the \emph{initial state}
and $\goal \in (\dom^+)^n$ is the \emph{goal}. 
Each action $a \in \acts$ has 
a \emph{precondition} $\pre(a) \in (\dom^+)^n$ and
an \emph{effect} $\eff(a) \in (\dom^+)^n$.
We will frequently use the convention that a variable has value $\undef$
in a precondition/effect unless a value is explicitly specified.
Let $a \in \acts$ and let $s \in \dom^n$.
Then $a$ is \emph{valid in $s$} if for all $v \in \vars$,
either $\proj{\pre(a)}{v} = \proj{s}{v}$ or $\proj{\pre(a)}{v} = \undef$.
Furthermore, the \emph{result of $a$ in $s$} is a state  $t \in \dom^n$
defined such that for all $v \in \vars$,
 $\proj{t}{v} = \proj{\eff(a)}{v}$ if $\proj{\eff(a)}{v} \neq \undef$
and $\proj{t}{v} = \proj{s}{v}$ otherwise.

Let $s_0, s_\ell \in \dom^n$ and 
let $\plan = \seq{a_1,\ldots,a_\ell}$ be a sequence of actions.
Then $\plan$ is a \emph{plan from $s_0$ to $s_\ell$} if
either 
(i)~$\plan = \seq{}$ and $\ell = 0$ or
(ii)~there are states $s_1,\ldots,s_{\ell-1} \in \dom^n$
such that for all $i$, where $1 \leq i \leq \ell$,
$a_i$ is valid in $s_{i-1}$ and $s_i$ is the result of $a_i$ in $s_{i-1}$.
A state $s \in \dom^n$ is a \emph{goal state}
if for all $v \in \vars$,
either $\proj{\goal}{v} = \proj{s}{v}$ or
$\proj{\goal}{v} = \undef$.
An action sequence $\plan$ is a \emph{plan for $\iplan$} if
it is a plan from~$\init$ to some goal state 
$s \in \dom^n$.
We will study the following problem:

\smallskip

\begin{quote}
\noindent
\BPE\\
\textit{Instance:}
A tuple $\tuple{\iplan,k}$ where $\iplan$ is a \sasplus\ 
instance and $k$ is a positive integer.\\
\textit{Parameter:}
The integer $k$.\\
\textit{Question:}
Does $\iplan$ have a plan of length at most $k$?
\end{quote}

\smallskip

\noindent
We will consider the following four syntactical restrictions,
originally defined by B\"{a}ckstr\"{o}m and Klein
\cite{BackstromKlein91}.
\begin{quote}
\begin{description}
  \item[P] (postunique):
    For each $v \in \vars$ and each $x \in \dom$
    there is at most one $a \in \acts$ such that
    $\proj{\eff(a)}{v} = x$.
  \item[U] (unary):
    For each $a \in \acts$,  $\proj{\eff(a)}{v} \neq \undef$ 
    for exactly one $v \in \vars$.
  \item[B] (Boolean):
    $\card{\dom} = 2$.
  \item[S] (single-valued):
    For all $a,b \in \acts$ and all $v \in \vars$,
    if 
    $\proj{\pre(a)}{v} \neq \undef$,
    $\proj{\pre(b)}{v} \neq \undef$ and
    $\proj{\eff(a)}{v} = \proj{\eff(b)}{v} = \undef$,
    then $\proj{\pre(a)}{v} = \proj{\pre(b)}{v}$.
\end{description}
\end{quote}

For any set $R$ of such restrictions we write \mbox{$R$-\BPE} to denote
the restriction of \BPE\ to only instances satisfying the restrictions
in $R$.  Additionally we will consider restrictions on the number of
preconditions and effects as previously considered
in~\cite{Bylander94}. For two non-negative integers $p$ and $e$ we write
\mbox{$(p,e)$-\BPE} to denote the restriction of \BPE{} to only
instances where every action has at most $p$ preconditions and at most
$e$ effects. Table~\ref{table:bylander} and
Figure~\ref{fig:pubs-lattice} summarize results from
\cite{Bylander94,BackstromNebel95,BackstromChenJonssonOrdyniakSzeider12}
combined with the results presented in this paper.

\section{Parameterized  Complexity of $(0,e)$-\BPE{}}

\sloppypar In this section we completely characterize the parameterized complexity
of \BPE{} for planning instances without preconditions. It is
known~\cite{BackstromChenJonssonOrdyniakSzeider12} that \BPE{} without
preconditions is contained in the parameterized complexity class
\W{1}. Here we show that $(0,e)$-\BPE{} is also \W{1}-hard for every $e
> 2$ but it becomes fixed-parameter tractable if $e \leq 2$.  Because
$(0,1)$-\BPE{} is trivially solvable in polynomial time this completely
characterized the parameterized complexity of \BPE{} without
preconditions.

\subsection{Hardness Results}
\begin{THE}
  $(0,3)$-\BPE{} is $\W{1}$-hard.
\end{THE}
\begin{proof}
  We devise a parameterized reduction from the following problem, which
  is $\W{1}$-complete~\cite{Pietrzak03}.

  \begin{quote}
    \textsc{Multicolored Clique}
    
    \emph{Instance:} A $k$\hy partite graph $G=(V,E)$ with partition
    $V_1,\dots,V_k$ such that $\Card{V_i}=\Card{V_j}=n$ for $1\leq i<j
    \leq k$.
    
    \emph{Parameter:} The integer $k$.  

    \emph{Question:} Are there vertices $v_1,\dots,v_k$ such that
    $v_i\in V_i$ for $1\leq i \leq k$ and $\{v_i,v_j\}\in E$ for $1\leq
    i < j \leq k$? (The graph $K=(\{v_1,\dots,v_k\},\SB \{v_i,v_j\} \SM
    1\leq i < j \leq k\SE)$ is a \emph{$k$-clique} of $G$.)
  \end{quote}
  Let $\insti=(G,k)$ be an instance of this problem with partition
  $V_1,\dots,V_k$, $\Card{V_1}=\dots=\Card{V_k}=n$ and parameter $k$.
  We construct a $(0,3)$-\BPE{} instance
  $\insti'=(\iplan',k')$ with
  $\iplan'=\tuple{\vars',\dom',\acts',\init',\goal'}$ such that
  $\insti$ is a \textsc{Yes}-instance if and only if so is $\insti'$.

  We set $\vars'=V(G) \cup \SB p_{i,j} \SM 1
  \leq i < j \leq k \SE$, $\dom'=\{0,1\}$, $\init'=\tuple{0,\ldots,0}$,
  $\proj{\goal'}{p_{i,j}}=1$ for every $1 \leq i < j \leq
  k$ and $\proj{\goal'}{v}=0$ for every $v \in V(G)$. 
  Furthermore, the set $\acts'$ contains the following actions:
  \begin{itemize}
  \item For every $v \in V(G)$ one action $a_v$ with $\proj{\eff(a_v)}{v}=0$;
  \item For every $e=\{v_i,v_j\} \in E(G)$ with $v_i \in V_i$ and $v_j
    \in V_j$ one action $a_e$ with $\proj{\eff(a_e)}{v_i}=1$,
    $\proj{\eff(a_e)}{v_j}=1$, and $\proj{\eff(a_e)}{p_{i,j}}=1$. 
  \end{itemize}
  Clearly, every action in $\acts'$ has no precondition and at most
  $3$ effects.  

  The theorem will follow after we have shown the that
  $G$ contains a $k$-clique if and only if $\iplan$ has a plan of
  length at most $k'=\binom{k}{2}+k$.
  Suppose that $G$ contains a $k$-clique with vertices $v_1,\dots,v_k$
  and edges $e_1, \dots, e_{k''}$, $k''=\binom{k}{2}$.  Then
  $\plan'=\seq{a_{e_1},\dots,a_{e_{k''}},a_{v_1},\dots,a_{v_k}}$ is a
  plan of length $k'$ for $\iplan'$.
  For the reverse direction suppose that $\plan'$ is a plan of length
  at most $k'$ for $\iplan'$. Because $\proj{\init'}{p_{i,j}}=0 \neq
  \proj{\goal'}{p_{i,j}}=1$ the plan $\plan'$ has to contain at least
  one action $a_e$ where $e$ is an edge between a vertex in $V_i$ and
  a vertex in $V_j$ for every $1 \leq i < j \leq k$. Because
  $\proj{\eff(a_{e=\{v_i,v_j\}})}{v_i}=1 \neq \proj{\goal}{v_i}=0$ and
  $\proj{\eff(a_{e=\{v_i,v_j\}})}{v_j}=1 \neq \proj{\goal}{v_j}=0$ for
  every such edge $e$ it follows that $\plan'$ has to contain at least
  one action $a_{v}$ with $v \in V_i$ for every $1 \leq i \leq
  k$. Because $k'=\binom{k}{2}+k$ it follows that $\plan'$ contains
  exactly $\binom{k}{2}$ actions of the form $a_e$ for some edge $e
  \in E(G)$ and exactly $k$ actions of the form $a_v$ for some vertex
  $v \in V(G)$. It follows that the graph $K=(\SB v \SM a_v \in \plan
  \SE,\SB e \SM a_e \in \plan \SE)$ is a $k$-clique of~$G$.  
\end{proof}

\subsection{Fixed-Parameter Tractability}\label{sec:res-fpt}

Before we show that $(0,2)$-\BPE{} is fixed-parameter tractable we
need to introduce some notions and prove some simple properties of
$(0,2)$-\BPE{}.
Let $\iplan=\tuple{\vars,\dom,\acts,\init,\goal}$ be an instance of \BPE{}.
We say an action
$a \in \acts$ has an effect on some variable $v \in \vars$ if
$\proj{\eff(a)}{v}\neq \undv$, we call this effect \emph{good} if
furthermore $\proj{\eff(a)}{v}=\proj{\goal}{v}$ or
$\proj{\goal}{v}=\undv$ and we call the effect \emph{bad}
otherwise. We say an action $a \in \acts$ is \emph{good} if it has
only good effects, \emph{bad} if it has only bad effects, and
\emph{mixed} if it has at least one good and at least one bad effect.
Note that if a valid plan contains a bad action then this
action can always be removed without changing the validity of the plan.
Consequently, we only need to
consider good and mixed actions.
Furthermore, we denote by $B(\vars)$ the set of
variables $v \in \vars$ with $\proj{\goal}{v}\neq\undv$ and 
$\proj{\init}{v} \neq \proj{\goal}{v}$. 
 
The next lemma shows that we do not need to consider good actions with
more than~$1$ effect for $(0,2)$-\BPE{}. 
\shortversion{
\begin{LEM}[$\star$]\label{lem:nodoublegoodactions}
  Let $\insti=\tuple{\iplan,k}$ be an instance of $(0,2)$-\BPE{}. Then
  $\insti$ can be fpt-reduced to an instance
  $\insti'=\tuple{\iplan',k'}$ of $(0,2)$-\BPE{} where $k'=k(k+3)+1$ and
  no good action of $\insti'$ effects more than one variable.
\end{LEM}}%
\longversion{
\begin{LEM}\label{lem:nodoublegoodactions}
  Let $\insti=\tuple{\iplan,k}$ be an instance of $(0,2)$-\BPE{}. Then
  $\insti$ can be fpt-reduced to an instance
  $\insti'=\tuple{\iplan',k'}$ of $(0,2)$-\BPE{} where $k'=k(k+3)+1$ and
  no good action of $\insti'$ effects more than one variable.
\end{LEM}
\begin{proof}
  The required instance $\insti'$ is constructed from $\insti$ as follows. 
  $\vars'$ contains the following variables:
  \begin{itemize}
  \item All variables in $\vars$;
  \item One binary variable $g$;
  \item For every action $a \in \acts$ and every $1 \leq i \leq
    k+2$ one binary variable $v_i(a)$;
  \end{itemize}
  $\acts'$ contains the following actions:
  \begin{itemize}
  \item For every mixed action $a \in \acts$ that has a good effect on
    the variable $v$ and a bad effect on the variable $v'$ one action
    $a_1(a)$ such that
    $\proj{\eff(a_1(a))}{v'}=\proj{\eff(a)}{v'}$ and
    $\proj{\eff(a_1(a))}{v_1(a)}=0$, one action $a_i(a)$ for every $1 < i <
    k+3$ such that
    $\proj{\eff(a_i(a))}{v_{i-1}(a)}=1$ and
    $\proj{\eff(a_i(a))}{v_i(a)}=0$, as well as one action $a_{k+3}(a)$ such
    that $\proj{\eff(a_{k+3}(a))}{v_{k+2}(a)}=1$ and
    $\proj{\eff(a_{k+3}(a))}{v}=\proj{\eff(a)}{v}$;
  \item  For every good action $a \in \acts$ that has only one effect
    on the variable $v$ one action
    $a_1$ such that
    $\proj{\eff(a_1(a))}{g}=1$ and
    $\proj{\eff(a_1(a))}{v_1(a)}=0$, one action $a_i(a)$ for every $1 < i <
    k+3$ such that
    $\proj{\eff(a_i(a))}{v_{i-1}(a)}=1$ and
    $\proj{\eff(a_i(a))}{v_i(a)}=0$, as well as one action $a_{k+3}(a)$ such
    that $\proj{\eff(a_{k+3}(a))}{v_{k+2}(a)}=1$ and
    $\proj{\eff(a_{k+3}(a))}{v}=\proj{\eff(a)}{v}$;
  \item For every good action $a \in \acts$ that has two effects
    on the variables $v$ and $v'$ one action
    $a_1(a)$ such that
    $\proj{\eff(a_1(a))}{g}=1$ and
    $\proj{\eff(a_1(a))}{v_1(a)}=0$, one action $a_i(a)$ for every $1 < i <
    k+2$ such that
    $\proj{\eff(a_i(a))}{v_{i-1}(a)}=1$ and
    $\proj{\eff(a_i(a))}{v_i(a)}=0$, one action $a_{k+2}(a)$ such
    that $\proj{\eff(a_{k+2}(a))}{v_{k+1}(a)}=1$ and
    $\proj{\eff(a_{k+2}(a))}{v}=\proj{\eff(a)}{v}$, as well as one action
    $a_{k+3}(a)$ such
    that $\proj{\eff(a_{k+3}(a))}{v_{k+1}(a)}=1$ and
    $\proj{\eff(a_{k+3}(a))}{v'}=\proj{\eff(a)}{v'}$;
  \item One action $a_g$ with $\proj{\eff(a_g)}{g}=0$.
  \end{itemize}
  We set $\dom'=\dom \cup \{0,1\}$, $\proj{\init'}{v}=\proj{\init}{v}$
  for every $v \in \vars$, $\proj{\init'}{v}=0$ for every $v
  \in \vars' \setminus \vars$, $\proj{\goal'}{v}=\proj{\goal}{v}$
  for every $v \in \vars$, $\proj{\goal'}{v}=0$ for every $v
  \in \vars' \setminus \vars$, and $k'=k(k+2)+1$.

  Clearly, $\insti'$ can be constructed from $\insti$ by an algorithm
  that is fixed-parameter tractable (with respect to $k$) and $\insti'$ is an
  instance of $(0,2)$-\BPE{} where no good
  action effects more than $1$ variable. It remains to show that
  $\insti'$ is equivalent to $\insti$.

  Suppose that $\plan=\seq{a_1,\dotsc,a_l}$ is a plan of length at
  most $k$ for $\iplan$. Then\\
  $\seq{a_{k+3}(a_1),\dotsc,a_{1}(a_1),\dotsc,
    a_{k+3}(a_l),\dotsc,a_{1}(a_l),a_g}$ is a plan of length $l(k+3)+1\leq
  k(k+3)+1$ for $\iplan'$.

  To see the reverse direction suppose that
  $\plan'=\seq{a_1,\dots,a_{l'}}$ is a minimal (with respect to
  sub-sequences) plan of length at most $k'$ for $\iplan'$. We say that
  $\plan'$ \emph{uses} an action $a \in \acts$ if $a_i(a) \in \plan'$
  for some $1 \leq i \leq k+3$. We also define an order of the actions
  used by $\plan'$ in the natural way, i.e., for two actions $a,a' \in
  \acts$ that are used by $\plan'$ we say that $a$ is smaller than $a'$
  if the first occurrence of an action $a_i(a)$ (for some $1 \leq i
  \leq k+3$) in~$\plan'$ 
  is before the first occurrence of an action
  $a_i(a')$ (for some $1 \leq i \leq k+3$) in~$\plan'$.
  
  Let $\plan=\seq{a_1,\dotsc,a_{l}}$ be the (unique) sequence
  of actions in $\acts$ that are
  used by $\plan'$ whose order corresponds to the order in which there
  are used by $\plan'$. Clearly, $\plan$ is a plan for $\iplan$. 
  It remains to show that $l\leq k$ for which we need the following
  claim.
  \begin{CLM}\label{APPENDIX-clm:many-or-no-actions}
    If $\plan'$ uses some action $a \in \acts$ then $\plan'$
    contains at least $k+2$ actions from $a_1(a),\dotsc,a_{k+3}(a)$.
  \end{CLM}
  Let $i$ be the largest integer with $1 \leq i \leq
  k+3$ such that $a_i(a)$ occurs in~$\plan'$. We first show by
  induction on $i$ that
  $\plan'$ contains all actions in $\SB a_j(a) \SM 1 \leq j \leq i
  \SE$. Clearly, if $i=1$ there is nothing to show, so assume that
  $i>1$. The induction step follows from the fact that the action $a_i(a)$ has a
  bad effect on the variable $v_{i-1}(a)$ and the action
  $a_{i-1}(a)$ is the only action of $\iplan'$ that has a good effect on
  $v_{i-1}(a)$ and hence $\plan'$ has to contain the action
  $a_{i-1}(a)$. It remains to show that $i \geq k+2$.
  Suppose for a contradiction that $i < k+2$ and consequently the
  action $a_{i+1}(a)$ is not contained in~$\plan'$. Because the action
  $a_{i+1}(a)$ is the only action of $\iplan'$
  that has a bad effect on the variable $v_i(a)$ it follows that the
  variable $v_i(a)$ remains in the goal state over the whole execution
  of the plan $\plan'$. But then $\plan'$ without the action $a_i(v)$
  would still be a plan for $\iplan'$ contradicting our assumption
  that $\iplan'$ is minimal with respect to sub-sequences. 

  It follows from Claim~\ref{APPENDIX-clm:many-or-no-actions} that $\plan'$ uses at most
  $\frac{l'}{k+2}\leq \frac{k'}{k+2}=\frac{k(k+3)+1}{k+2}<k+1$ actions
  from $\acts$. Hence, $l \leq k$ proving the lemma.
\end{proof}
}
\begin{THE}\label{the:02fpt}
  $(0,2)$-\BPE{} is fixed-parameter tractable.
\end{THE}
\begin{proof}
  We show fixed-parameter tractability of $(0,2)$-\BPE{} by reducing it
  to the following fixed-parameter tractable problem~\cite{GuoNiedermeierSuchy11}.
  \smallskip
  \begin{quote}
    \noindent
    \textsc{Directed Steiner Tree}\\
    \noindent  
    \emph{Instance:} A set of nodes $N$, 
    a weight function $w\ :\  N \times N \rightarrow (\mathbb{N} \cup
    \{\infty\})$, 
    a root node $s \in N$,
    a set $T \subseteq N$ of terminals , and a weight bound $p$.\\
    \noindent  
    \emph{Parameter:} $p_M=\frac{p}{\min\SB w(u,v) \SM u,v \in N\SE}$.\\  
    \noindent
    \emph{Question:} Is there a set of arcs $E \subseteq N \times N$ of
    weight $w(E) \leq p$ (where $w(E)=\sum_{e \in E}w(e)$) such that in
    the digraph $D=(N,E)$ for every $t \in T$ there is a directed path
    from $s$ to $t$? We will call the digraph $D$ a \emph{directed
      Steiner Tree (DST)} of weight $w(E)$.
  \end{quote}
  \smallskip Let $\insti=\tuple{\iplan,k}$ where
  $\iplan=\tuple{\vars,\dom,\acts,\init,\goal}$ be an instance of
  $(0,2)$-\BPE{}. Because of Lemma~\ref{lem:nodoublegoodactions} we can
  assume that $\acts$ contains no good actions with two effects.  We
  construct an instance $\insti'=\tuple{N,w,s,T,p}$ of \textsc{Directed
    Steiner Tree} where $p_M=k$ such that $\insti$ is a
  \textsc{Yes}-instance if and only if $\insti'$ is a
  \textsc{Yes}-instance. Because $p_M=k$ this shows that $(0,2)$-\BPE{}
  is fixed-parameter tractable.

  We are now ready to define
  the instance $\insti'$. 
  The node set $N$ 
  consists of the root vertex $s$
  and one node for every variable in $\vars$. The weight function
  $w$ is $\infty$ for all but the following arcs:
  \longversion{\par}
  (i)~For every good action $a \in \acts$ the arc from $s$ to the
    unique variable $v \in \vars$ that is effected by $a$ gets weight~$1$.
  \longversion{\par}
    (ii) For every mixed action $a \in \acts$ with some good effect on
    some variable $v_g \in \vars$ and some bad effect on some variable
    $v_b \in \vars$, the arc from $v_b$ to $v_g$ gets weight $1$.

  We identify the root $s$ from the instance $\insti$ 
  with the node $s$, we let $T$ be the set $B(\vars)$, and
  $p_M=p=k$.
\shortversion{
  \begin{CLM}[$\star$]\label{claim:dst}
    $\iplan$ has a plan of length at most $k$ if and only if $\insti'$
    has a DST of weight at most $p_M=p=k$.
  \end{CLM}
The theorem follows.
}
\longversion{
  \begin{CLM}\label{claim:dst}
    $\iplan$ has a plan of length at most $k$ if and only if $\insti'$
    has a DST of weight at most $p_M=p=k$.
  \end{CLM}
  Suppose $\iplan$ has a plan $\plan=\seq{a_1,\dots,a_l}$ with $l
  \leq k$. W.l.o.g.~we can assume that $\plan$ contains no bad
  actions.
  The arc set $E$ that corresponds to $\plan$ consists of
  the following arcs:
  \longversion{\par}
(i)~For every good action $a \in \plan$ that has its unique good
    effect on a variable $v \in \vars$, the set $E$ contains the arc
    from $s$ to $v$.
  \longversion{\par}
(ii)~For every mixed action $a \in \plan$ with a good effect on some
    variable $v_g$ and a bad effect on some variable $v_b$, the set $E$ contains an arc
    from $v_b$ to $v_g$.

  It follows that the weight of $E$ equals the number of
  actions in~$\plan$ and hence is at most $p=k$ as
  required. It remains to show that the digraph $D=(V,E)$ is a
  DST, i.e., $D$ contains a directed path from the vertex $s$ to
  every vertex in $T$. 
  Suppose to the contrary that there is a terminal
  $t \in T$ that is not reachable from $s$ in~$D$. Furthermore, let $R
  \subseteq E$ be the set of all arcs in $E$ such that $D$ contains a
  directed path from the tail of every arc in $R$ to $t$. It follows
  that no arc in $R$ is incident to $s$. Hence, $R$ only consists of
  arcs that correspond to mixed actions in~$\plan$. If $R=\emptyset$
  then the plan $\plan$ does not contain an action that effects the
  variable $t$. But this contradicts our assumption that $\plan$ is a
  plan (because $t \in B(\vars)$). Hence, $R \neq \emptyset$. Let $a$ be
  the mixed action corresponding to an arc in $R$ that occurs last in
  $\plan$ (among all mixed actions that correspond to an arc in
  $R$). Furthermore, let $v \in \vars$ be the variable that is badly
  affected by $a$. Then $\plan$ can not be a plan because after the
  occurrence of $a$ in~$\plan$ there is no action in~$\plan$ that
  affects $v$ and hence $v$ can not be in the goal state after $\plan$
  is executed.

  To see the reverse direction, let $E \subseteq N \times N$ be a
  solution of $\insti$ and let $D=(N,E)$ be the {DST}.  W.l.o.g.~we can
  assume that $D$ is a directed acyclic tree rooted in $s$ (this follows
  from the minimality of $D$). We obtain a plan $\plan$ of length at
  most $p$ for $\iplan$ by traversing the DST $D$ in a bottom-up
  manner. More formally, let $d$ be the maximum distance from $s$ to any
  node in $T$, and for every $1 \leq i < d$ let $A(i)$ be the set of
  actions in $\acts$ that correspond to arcs in $E$ whose tail is at
  distance $i$ from the node $s$. Then $\plan=\seq{A(d-1),\dots,A(1)}$
  (for every $1 \leq i \leq d-1$ the actions contained in $A(d-1)$ can
  be executed in an arbitrary order) is a plan of length at most $k=p$
  for $\iplan$.
  
Hence Claim~\ref{claim:dst} is established, and the theorem follows.}
\end{proof}

\section{Kernel Lower Bounds}

\longversion{\subsection{Kernel Lower Bounds for $(0,2)$-\BPE{}}}

Since $(0,2)$-\BPE{} is fixed-parameter tractable by
Theorem~\ref{the:02fpt} it admits a kernel. Next we provide strong
theoretical evidence that the problem does not admit a polynomial
kernel. \shortversion{The proof of Theorem~\ref{02bpenokernel} is
  based on an OR-composition algorithm and
  Proposition~\ref{pro:or-comp-no-poly-kernel}.}

\shortversion{
\begin{THE}[$\star$] \label{02bpenokernel}
  $(0,2)$-\BPE{} has no polynomial kernel unless
  $\NOPOLYKERNEL$.
\end{THE}
}

\longversion{
\begin{THE} \label{02bpenokernel}
  $(0,2)$-\BPE{} has no polynomial kernel unless
  $\NOPOLYKERNEL$.
\end{THE}
\begin{proof}
  Because of Proposition~\ref{pro:or-comp-no-poly-kernel} it suffices
  to devise an OR-composition algorithm for $(0,2)$-\BPE{}.
  Suppose we are given
  $t$ instances
  $\insti_1=\tuple{\iplan_1,k},\dots,\insti_t=\tuple{\iplan_t,k}$ of
  $(0,2)$-\BPE{} where
  $\iplan_i=\tuple{\vars_i,\dom_i,\acts_i,\init_i,\goal_i}$ for every
  $1 \leq i \leq t$. We will now show how we can construct the
  required instance $\insti=\tuple{\iplan,k''}$ of $(0,2)$-\BPE{} via
  an OR-composition algorithm. As a first step we compute the
  instances
  $\insti_1'=\tuple{\iplan_1',k'},\dots,\insti_t'=\tuple{\iplan_t',k'}$
  from the instances
  $\insti_1=\tuple{\iplan_1,k},\dots,\insti_t=\tuple{\iplan_t,k}$
  according to Lemma~\ref{lem:nodoublegoodactions}. Then
  $\vars$ consists of the following variables:
  \longversion{\par}
  (i)~the variables $\bigcup_{1\leq i \leq t}\vars_i'$;
  \longversion{\par}
  (ii)~binary variables $b_1,\dotsc,b_{k'}$;
  \longversion{\par}
  (iii) for every $1 \leq i \leq t$ and $1 \leq j < 2k'$ a binary
    variable $p(i,j)$;
  \longversion{\par}
  (iv) A binary variable $r$.

  $\acts$ contains the action $a_r$ with $\proj{\eff(a_r)}{r}=0$ and
  the following additional actions for every $1 \leq i \leq t$:

  (i) The actions $\acts_i' \setminus a_g^i$, where $a_g^i$ is the
    copy of the action $a_g$ for the instance $\insti_i'$ (recall the
    construction of $\insti_i'$ given in Lemma~\ref{lem:nodoublegoodactions});

 (ii) An action $a_i(r)$ with $\proj{\eff(a_i(r))}{r}=1$ and
    $\proj{\eff(a_i(r))}{p(i,1)}=0$;

  (iii) For every $1 \leq j < 2k'-1$ an action
    $a_{i,j}$ with $\proj{\eff(a_{i,j})}{p_{i,j}}=1$ and
    $\proj{\eff(a_{i,j})}{p(i,j+1)}=0$;

 (iv) An action $a_i(g)$ with $\proj{\eff(a_i(g))}{p_{i,2k'-1}}=1$ and
    $\proj{\eff(a_i(g))}{g^i}=0$ where $g^i$ is the
    copy of the variable~$g$ for the instance $\insti_i'$ (recall the
    construction of $\insti_i'$ given in Lemma~\ref{lem:nodoublegoodactions});

 (v) Let $v_1,\dotsc,v_r$ for $r \leq k'$ be an
    arbitrary ordering of the variables in $B(\vars_i)$
    (recall the
    definition of $B(\vars_i)$ from Section~\ref{sec:res-fpt}). 
    Then for every $1 \leq j \leq r$ we introduce an action
    $a_i(b_j)$ with $\proj{\eff(a_i(b_j))}{v_j}=\proj{\init_i'}{v_j}$ and
    $\proj{\eff(a_i(b_j))}{b_j}=0$. Furthermore, for 
    every $r < j \leq k'$ we introduce an action $a_i(b_j)$ with
    $\proj{\eff(a_i(b_j))}{v_r}=\proj{\init_i'}{v_r}$ and
    $\proj{\eff(a_i(b_j))}{b_j}=0$.

  We set $\dom=\bigcup_{1 \leq i \leq t}\dom_i' \cup \{0,1\}$, $\proj{\init}{v}=\proj{\init_i'}{v}$
  for every $v \in \vars_i'$ and $1 \leq i \leq t$, 
  $\proj{\init}{v}=0$ for every $v \in \vars \setminus ((\bigcup_{1
    \leq i \leq t}\vars_i') \cup \{b_1,\dotsc,b_{k'}\})$, 
  $\proj{\init}{v}=1$ for every $v \in \{b_1,\dotsc,b_k\}$,
  $\proj{\goal}{v}=\proj{\init_i'}{v}$
  for every $v \in \vars_i'$ and $1 \leq i \leq t$, 
  $\proj{\goal}{v}=0$ for every $v \in \vars \setminus (\bigcup_{1
    \leq i \leq t}\vars_i')$, and $k''=4k'+1$.

  Clearly, $\insti$ can be constructed from
  $\insti_1,\dots,\insti_t$ in polynomial time with respect to
  $\sum_{1\leq i \leq t}|\insti_i|+k$
  and the parameter $k''=4k'+1=4(k(k+3)+1)+1$
  is polynomial bounded by the parameter $k$. 
  \longversion{By showing the following claim we conclude the proof of
    the theorem.} 
  \begin{CLM}\label{clm:or}
    $\insti$ is a \textsc{Yes}-instance if and
    only if at least one of the instances $\insti_1,\dots,\insti_t$ is
    a \textsc{Yes}-instance.  
  \end{CLM}

  Suppose that there is an $1 \leq i \leq t$ such that $\iplan_i$
  has a plan of length at most $k$. It follows from
  Lemma~\ref{lem:nodoublegoodactions} 
  that $\iplan_i'$ has a plan $\plan'$ of length at most $k'$. Then it
  is straightforward to check that $\plan=\seq{a_i(b_1), \dotsc,a_i(b_{k'})}
  \concat \plan' \concat \seq{a_i(g),a_{i,2k'-2},\dotsc,a_{i,1},a_i(r),a_r}$ is a
  plan of length at most $4k'+1$ for $\iplan$.

  For the reverse direction let $\plan$ be a plan of length at most
  $k''$ for $\iplan$. W.l.o.g.~we can assume that for every $1\leq i
  \leq t$ the set $B(\vars_i')$ is not empty and hence every
  plan for $\iplan_i'$ has to contain at least one good action $a \in
  \acts_i'$. Because $\proj{\eff(a)}{g^i} \neq \init_i'[g]$ for every such
  good action $a$ (recall the construction of $\insti_i'$ according to
  Lemma~\ref{lem:nodoublegoodactions}) it follows that there is an $1
  \leq i \leq t$ such that $\plan$ contains all the $2k'+1$ actions
  $a_i(g),a_{i,2k'-2},\dotsc,a_{i,1},a_i(r),a_r$. Furthermore, because
  $k''<2(2k'+1)$ there can be at most one such $i$ and hence
  $\plan \cap \bigcup_{1 \leq j \leq t}\acts_j' \subseteq
  \acts_i'$. Because $B(\vars)=\{b_1,\dotsc,b_{k'}\}$ the plan $\plan$
  also has to contain the actions
  $a_i(b_1),\dotsc,a_i(b_{k'})$. Because of the effects (on the
  variables in $B(V_i)$) of these
  actions it follows that $\plan$ has to contain a plan $\plan_i'$ of
  length at most $4k'+1-(2k'+1)-k'=k'$ for $\iplan_i'$. It now follows
  from Lemma~\ref{lem:nodoublegoodactions} that $\iplan_i$ has a plan
  of length at most $k$.
\end{proof}
}
\longversion{\subsection{Kernel Lower Bounds for PUBS Restrictions}}

\sloppypar In previous work~\cite{BackstromChenJonssonOrdyniakSzeider12}
we have classified the parameterized complexity of the ``PUBS''
fragments of \BPE{}. It turned out that the problems fall into four
categories  (see Figure~\ref{fig:pubs-lattice}):
  \longversion{\par}
 (i)~polynomial-time solvable, 
  \longversion{\par}
(ii)~NP-hard but fixed-parameter tractable, 
  \longversion{\par}
(iii)~$\W{1}$\hy complete, and
  \longversion{\par}
(iv)~$\W{2}$\hy complete.
  \longversion{\par}
  The aim
of this section is to further refine this classification with respect to
kernelization. The problems in category~(i) trivially admit a kernel of
constant size, whereas the problems in categories~(iii) and (iv) do not admit a
kernel at all (polynomial or not), unless $\W{1}=\FPT$ or $\W{2}=\FPT$,
respectively. Hence it remains to consider the six problems in
category~(ii), each of them could either admit a polynomial kernel or
not. We show that none of them does.
 
According to our classification
\cite{BackstromChenJonssonOrdyniakSzeider12}, the problems in
category~(ii) are exactly the problems $R$-\BPE{}, for $R \subseteq
\{P,U,B,S\}$, such that $P\in R$ and $\{P,U,S\} \not\subseteq R$.

\begin{THE}\label{thm:nopolykernel-p}
  None of the problems $R$-\BPE{} for $R \subseteq \{P,U,B,S\}$ such
  that $P\in R$ and $\{P,U,S\} \not\subseteq R$ (i.e., the problems in
  category~(ii)) admits a polynomial kernel unless $\NOPOLYKERNEL$.
\end{THE}

The remainder of this section is devoted to
establish~Theorem~\ref{thm:nopolykernel-p}.  The relationship between
the problems as indicated in Figure~\ref{fig:pubs-lattice} greatly
simplifies the proof.  Instead of considering all six problems
separately, we can focus on the two most restricted problems
$\{P,U,B\}$-\BPE{} and $\{P,B,S\}$-\BPE{}.  If any other problem in
category~(ii) would have a polynomial kernel, then at least one of these
two problems would have one. This follows by
Proposition~\ref{pro:poly-par-reduction} and the following facts:
\begin{enumerate}
\item The unparameterized versions of all the problems in category~(ii)
  are \NP-complete. This holds since  the corresponding classical
  problems are strongly \NP-hard, hence the problems remain \NP-hard
  when $k$ is encoded in unary (as shown by B\"{a}ckstr\"{o}m and Nebel
  \cite{BackstromNebel95});
\item If $R_1\subseteq R_2$ then the identity function gives a
  polynomial parameter reduction from $R_2$-\BPE{} to $R_1$-\BPE{}.
\end{enumerate}
Furthermore, the following result of B\"{a}ckstr\"{o}m and Nebel
\cite[Theorem 4.16]{BackstromNebel95} even provides a polynomial
parameter reduction from $\{P,U,B\}$-\BPE{} to
$\{P,B,S\}$-\BPE{}. Consequently, $\{P,U,B\}$-\BPE{} remains the only
problem for which we need to establish a superpolynomial kernel lower
bound.

\begin{PRO}[B\"{a}ckstr\"{o}m and Nebel \cite{BackstromNebel95}]
  \label{pro:from-pb-to-pbs}
  Let $\insti=\tuple{\iplan,k}$ be an instance of
  $\{P,U,B\}$-\BPE{}. Then $\insti$ can be transformed in polynomial
  time into an equivalent instance $\insti'=\tuple{\iplan',k'}$ of
  $\{P,B,S\}$-\BPE{} such that $k=k'$.
\end{PRO}
Hence, in order to complete the proof of
Theorem~\ref{thm:nopolykernel-p} it only remains to establish the next
lemma.

\begin{LEM}\label{lem:nopolykernel-pub}
  $\{P,U,B\}$-\BPE{}
 has no polynomial kernel unless $\NOPOLYKERNEL$.
\end{LEM}
\begin{proof}
  Because of Proposition~\ref{pro:or-comp-no-poly-kernel}, it suffices
  to devise an OR-composition algorithm for $\{P,U,B\}$-\BPE{}.  Suppose
  we are given $t$ instances
  $\insti_1=\tuple{\iplan_1,k},\dots,\insti_t=\tuple{\iplan_t,k}$ of
  $\{P,U,B\}$-\BPE{} where
  $\iplan_i=\tuple{\vars_i,\dom_i,\acts_i,\init_i,\goal_i}$ for every $1
  \leq i \leq t$.  
  It has been shown in~\cite[Theorem~5]{BackstromChenJonssonOrdyniakSzeider12} that
  $\{P,U,B\}$-\BPE{} can be solved in time
  $O^*(S(k))$ (where $S(k)=2 \cdot 2^{(k+2)^2} \cdot (k+2)^{(k+1)^2}$ 
  and the $O^*$ notation suppresses
  polynomial factors).
  It follows that $\{P,U,B\}$-\BPE{} can be solved in polynomial time
  with respect to $\sum_{1 \leq i \leq t}|\insti_i|+k$ if
  $t>S(k)$. Hence, if $t>S(k)$ this gives us an OR-composition algorithm
  as follows. We first run the algorithm for $\{P,U,B\}$-\BPE{} on each
  of the $t$ instances. If one of these $t$ instances is a
  \textsc{Yes}-instance then we output this instance. If not then we
  output any of the $t$ instances. This shows that $\{P,U,B\}$-\BPE{}
  has an OR-composition algorithm for the case that $t>S(k)$. Hence, in
  the following we can assume that $t \leq S(k)$.

  Given $\insti_1,\dots,\insti_t$ we will construct an instance
  $\insti=\tuple{\iplan,k'}$ of $\{P,U,B\}$-\BPE{} as follows. For the
  construction of $\insti$ we need the following auxiliary gadget, which will be
  used to calculate the logical ``OR'' of two binary variables. The
  construction of the gadget uses ideas from~\cite[Theorem 4.15]{BackstromNebel95}. Assume
  that $v_1$ and $v_2$ are two binary variables. The
  gadget $\textup{OR}_2(v_1,v_2,o)$ consists of the five binary
  variables $o_1$, $o_2$, $o$, $i_1$, and $i_2$. Furthermore,
  $\textup{OR}_2(v_1,v_2,o)$ contains the following actions:
  \begin{itemize}
  \item the action $a_o$ with
    $\proj{\pre(a_o)}{o_1}=\proj{\pre(a_o)}{o_2}=1$ and $\proj{\eff(a_o)}{o}=1$;
  \item the action $a_{o_1}$ with
    $\proj{\pre(a_{o_1})}{i_1}=1$, $\proj{\pre(a_{o_1})}{i_2}=0$ and
    $\proj{\eff(a_{o_1})}{o_1}=1$;
  \item the action $a_{o_2}$ with
    $\proj{\pre(a_{o_2})}{i_1}=0$, $\proj{\pre(a_{o_2})}{i_2}=1$ and $\proj{\eff(a_{o_2})}{o_2}=1$;
  \item the action $a_{i_1}$ with
    $\proj{\eff(a_{i_1})}{i_1}=1$;
  \item the action $a_{i_2}$ with
    $\proj{\eff(a_{i_2})}{i_2}=1$;
  \item the action $a_{v_1}$ with $\proj{\pre(a_{v_1})}{v_1}=1$ and
    $\proj{\eff(a_{v_1})}{i_1}=0$;
  \item the action $a_{v_2}$ with $\proj{\pre(a_{v_2})}{v_2}=1$ and
    $\proj{\eff(a_{v_2})}{i_2}=0$;
  \end{itemize}
  We now show that $\textup{OR}_2(v_1,v_2,o)$ can indeed be used to
  compute the logical ``OR'' of the variables $v_1$ and $v_2$. We need
  the following claim.
\shortversion{\begin{CLM}[$\star$]\label{clm:pub-nokernel-2or}
    Let $\iplan(\textup{OR}_2(v_1,v_2,o))$ be a $\{P,U,B\}$-\BPE{} instance
    that consists of the two binary variables $v_1$ and $v_2$, and the
    variables and actions of the gadget
    $\textup{OR}_2(v_1,v_2,o)$. Furthermore, let the initial state of 
    $\iplan(\textup{OR}_2(v_1,v_2,o))$ be any initial state that sets
    all variables of the gadget $\textup{OR}_2(v_1,v_2,o)$ to $0$ but
    assigns the variables $v_1$ and $v_2$ arbitrarily, and let the
    goal state of $\iplan(\textup{OR}_2(v_1,v_2,o))$ 
    be defined by $\proj{\goal}{o}=1$. Then
    $\iplan(\textup{OR}_2(v_1,v_2,o))$ has a plan if and only if its
    initial state sets at least one of the variables $v_1$ or $v_2$ to
    $1$. Furthermore, if there is such a plan then its length is $6$.
  \end{CLM} 
}
\longversion{\begin{CLM}\label{clm:pub-nokernel-2or}
    Let $\iplan(\textup{OR}_2(v_1,v_2,o))$ be a $\{P,U,B\}$-\BPE{} instance
    that consists of the two binary variables $v_1$ and $v_2$, and the
    variables and actions of the gadget
    $\textup{OR}_2(v_1,v_2,o)$. Furthermore, let the initial state of 
    $\iplan(\textup{OR}_2(v_1,v_2,o))$ be any initial state that sets
    all variables of the gadget $\textup{OR}_2(v_1,v_2,o)$ to $0$ but
    assigns the variables $v_1$ and $v_2$ arbitrarily, and let the
    goal state of $\iplan(\textup{OR}_2(v_1,v_2,o))$ 
    be defined by $\proj{\goal}{o}=1$. Then
    $\iplan(\textup{OR}_2(v_1,v_2,o))$ has a plan if and only if its
    initial state sets at least one of the variables $v_1$ or $v_2$ to
    $1$. Furthermore, if there is such a plan then its length is $6$.
  \end{CLM}
  To see the claim, suppose that there is a plan $\plan$ for
  $\iplan(\textup{OR}_2(v_1,v_2,o))$ and assume for a contradiction
  that both variables $v_1$ and $v_2$ are initially set to $0$. It is
  easy to see that the value of $v_1$ and $v_2$ can not change during
  the whole duration of the plan and that $\plan$ has to contain the
  actions $a_{o_1}$ and $a_{o_2}$. W.l.o.g.~we can assume that $\plan$
  contains $a_{o_1}$ before it contains $a_{o_2}$. Because of the
  preconditions of the actions $a_{o_1}$ and $a_{o_2}$, the variable
  $i_1$ must have value $1$ before $a_{o_1}$ occurs in~$\plan$ and it
  must have value $0$ before the action $a_{o_2}$ occurs in
  $\plan$. Hence, $\plan$ must contain an action that sets the
  variable $i_1$ to $0$. However, this can not be the case, since the
  only action setting $i_1$ to $0$ is the action $a_{v_1}$ which can
  not occur in~$\plan$ because the variable $v_1$ is $0$ for the whole
  duration of $\plan$.

  To see the reverse direction suppose that one of the variables $v_1$
  or $v_2$ is initially set to $1$. If $v_1$ is initially set to one
  then $\seq{a_{i_1},a_{o_1},a_{v_1},a_{i_2},a_{o_2},a_o}$ is a plan
  of length $6$ for $\iplan(\textup{OR}_2(v_1,v_2,o))$. On the other
  hand, if $v_2$ is initially set to one
  then $\seq{a_{i_2},a_{o_2},a_{v_2},a_{i_1},a_{o_1},a_o}$ is a plan
  of length $6$ for $\iplan(\textup{OR}_2(v_1,v_2,o))$. Hence the claim
  is shown true.}

  We continue by showing how
  we can use the gadget $\textup{OR}_2(v_1,v_2,o)$ to construct
  a gadget $\textup{OR}(v_1,\dots,v_r,o)$ such that there
  is a sequence of actions of $\textup{OR}(v_1,\dots,v_r,o)$ that
  sets the variable $o$ to $1$ if and only if at least one of the external
  variables $v_1,\dots,v_r$ are initially set to $1$. Furthermore, if
  there is such a sequence of actions then its length is at most $6
  \lceil \log r \rceil$. Let $T$ be a
  rooted binary tree with root $s$ that has $r$ leaves $l_1,\dots,l_r$ and is of
  smallest possible height. For every node $t \in V(T)$ we make a copy
  of our binary OR-gadget such that the copy of a leave node $l_i$ is
  the gadget $\textup{OR}_2(v_{2i-1},v_{2i},o_{l_i})$ and the
  copy of an inner node $t \in V(T)$ with children $t_1$ and $t_2$ is the gadget
  $\textup{OR}_2(o_{t_1},o_{t_2},o_{t})$ (clearly this needs to be adapted
  if $r$ is odd or an inner node has only one child). For the root
  node with children $t_1$ and $t_2$ the gadget becomes
  $\textup{OR}_2(o_{t_1},o_{t_2},o)$.
  This completes
  the construction of the gadget
  $\textup{OR}(v_1,\dots,v_r,o)$. Using
  Claim~\ref{clm:pub-nokernel-2or} it is easy to verify that the
  gadget $\textup{OR}(v_1,\dots,v_r,o)$ can indeed be used to compute
  the logical ``OR'' or the variables $v_1,\dotsc,v_r$.

  We are now ready to construct the instance $\insti$. $\insti$
  contains all the variables and actions from every instance
  $\insti_1,\dots,\insti_t$ and of the gadget
  $\textup{OR}(v_1,\dots,v_t,o)$. Additionally, $\insti$ contains the binary
  variables $v_1,\dots,v_t$ and the actions $a_1,\dots,a_t$ with
  $\pre(a_i)=\goal_i$ and $\proj{\eff(a_i)}{v_i}=1$.
  Furthermore, the initial state
  $\init$ of $\insti$ is defined as $\proj{\init}{v}=\proj{\init_i}{v}$
  if $v$ is a variable of $\insti_i$ and $\proj{\init}{v}=0$,
  otherwise. The goal state of $\insti$ is defined by
  $\proj{\goal}{o}=1$ and we set $k'=k+6\lceil \log t
  \rceil$. Clearly, $\insti$ can be constructed from
  $\insti_1,\dots,\insti_t$ in polynomial time
  and $\insti$ is a \textsc{Yes}-instance if and
  only if at least one of the instances $\insti_1,\dots,\insti_t$ is
  a \textsc{Yes}-instance. Furthermore, because $k'=k+6\lceil \log t \rceil\leq
  k+6\lceil \log S(k) \rceil=k+6\lceil 1+(k+2)^2 + (k+1)^2 \cdot \log (k+2)  \rceil$, the parameter $k'$
  is polynomial bounded by the parameter $k$. This concludes the proof
  of the lemma.
\end{proof}

\section{Conclusion}

We have studied the parameterized complexity of \BPE{} with respect to
the parameter plan length. In particular, we have shown that
$(0,e)$-\BPE{} is fixed-parameter tractable for $e\leq 2$ and
$\W{1}$-complete for $e>2$.  Together with our previous
results~\cite{BackstromChenJonssonOrdyniakSzeider12} this completes the
full classification of planning in Bylander's system of restrictions
(see Table~\ref{table:bylander}). Interestingly, $(0,2)$-\BPE{} turns
out to be the only nontrivial fixed-parameter tractable case (where the
unparameterized version is NP-hard).

We have also provided a full classification of kernel sizes for
$(0,2)$-\BPE{} and all the fixed-parameter tractable fragments of \BPE{}
in the ``PUBS'' framework. It turns out that none of the nontrivial
problems (where the unparameterized version is NP-hard) admits a
polynomial kernel unless the Polynomial Hierarchy collapses. This
implies an interesting \emph{dichotomy} concerning the kernel size: we
only have constant-size and superpolynomial kernels---polynomially
bounded kernels that are not of constant size are absent.

\longversion{
\bibliographystyle{abbrv}
\bibliography{literature}
}
\shortversion{

}

\end{document}